\documentclass[sigconf, nonacm]{acmart}

\usepackage{balance}
\usepackage{amsthm}
\usepackage{algorithm}
\usepackage{algpseudocode}
\usepackage{amsmath}
\usepackage{multirow}
\usepackage{diagbox}
\usepackage{url}
\usepackage{epstopdf}
\usepackage{caption}
\usepackage{graphicx, subfig}

\newcommand\vldbdoi{XX.XX/XXX.XX}
\newcommand\vldbpages{XXX-XXX}
\newcommand\vldbvolume{16}
\newcommand\vldbissue{1}
\newcommand\vldbyear{2023}
\newcommand\vldbauthors{\authors}
\newcommand\vldbtitle{\shorttitle}
\newcommand\vldbavailabilityurl{URL_TO_YOUR_ARTIFACTS}
\newcommand\vldbpagestyle{plain}

\begin{document}

\title{An Improved Christofides Mechanism for Local Differential Privacy Framework}

\author{She Sun}
\affiliation{%
  \institution{AI Research Institute\\ Zhejiang Lab}
  \city{Hangzhou}
  \state{Zhejiang}
  \country{China}}
\email{sunshe@zhejianglab.com}

\author{Li Zhou}
\affiliation{%
  \institution{AI Research Institute\\ Zhejiang Lab}
  \city{Hangzhou}
  \state{Zhejiang}
  \country{China}}
\email{zhou.li@zhejianglab.com}

\author{Xiaoran Yan}
\authornote{*Corresponding author: yanxr@zhejianglab.com.}
\affiliation{%
  \institution{AI Research Institute\\ Zhejiang Lab}
  \city{Hangzhou}
  \state{Zhejiang}
  \country{China}}
\email{yanxr@zhejianglab.com}

\begin{abstract}

The development of Internet technology enables an analysis on the whole population rather than a certain number of samples, and leads to increasing requirement for privacy protection. Local differential privacy (LDP) is an effective standard of privacy measurement; however, its large variance of mean estimation causes challenges in application. To address this problem, this paper presents a new LDP approach, an improved Christofides mechanism.

It compared four statistical survey methods for conducting surveys on sensitive topics---modified Warner, Simmons, Christofides, and the improved Christofides mechanism. Specifically, Warner, Simmons and Christofides mechanisms have been modified to draw a sample from the population without replacement, to decrease variance. Furthermore, by drawing cards without replacement based on modified Christofides mechanism, we introduce a new mechanism called the improved Christofides mechanism, which is found to have the smallest variance  under certain assumption when using LDP as a measurement of privacy leakage. The assumption is do satisfied usually in the real world. Actually, we decrease the variance to $28.7\%$ of modified Christofides mechanism's variance in our experiment based on the HCOVANY dataset---a real world dataset of IPUMS USA. This means our method gets a more accurate estimate by using LDP as a measurement of privacy leakage. This is the first time the improved Christofides mechanism is proposed for LDP framework based on comparative analysis of four mechanisms using LDP as the same measurement of privacy leakage.

\end{abstract}

\maketitle

\pagestyle{\vldbpagestyle}
\begingroup\small\noindent\raggedright\textbf{PVLDB Reference Format:}\\
\vldbauthors. \vldbtitle. PVLDB, \vldbvolume(\vldbissue): \vldbpages, \vldbyear.\\
\href{https://doi.org/\vldbdoi}{doi:\vldbdoi}
\endgroup
\begingroup
\renewcommand\thefootnote{}\footnote{\noindent
This work is licensed under the Creative Commons BY-NC-ND 4.0 International License. Visit \url{https://creativecommons.org/licenses/by-nc-nd/4.0/} to view a copy of this license. For any use beyond those covered by this license, obtain permission by emailing \href{mailto:info@vldb.org}{info@vldb.org}. Copyright is held by the owner/author(s). Publication rights licensed to the VLDB Endowment. \\
\raggedright Proceedings of the VLDB Endowment, Vol. \vldbvolume, No. \vldbissue\ %
ISSN 2150-8097. \\
\href{https://doi.org/\vldbdoi}{doi:\vldbdoi} \\
}\addtocounter{footnote}{-1}\endgroup

\ifdefempty{\vldbavailabilityurl}{}{
\vspace{.3cm}
\begingroup\small\noindent\raggedright\textbf{PVLDB Artifact Availability:}\\
The source code, data, and/or other artifacts have been made available at \url{\vldbavailabilityurl}.
\endgroup
}

\section{Introduction}

Local differential privacy (LDP), proposed in~\cite{raskhodnikova2008can} as the first equivalent definition and further developed by Duchi et al.~\cite{duchi2013local}, is an algorithm which quantifies privacy by having users randomly perturb their data locally and send the perturbed data to a (possibly un-trusted) data collector. Since the data collector does not hold the original personal data, it is a strong privacy model~\cite{yang2020answering, wang2020locally, wang2018locally, cormode2021frequency, bao2021cgm, cao2018quantifying, xiang2020linear, asi2022optimal, bassily2017practical}. LDP has been adopted by many industry organizations, including Apple~\cite{team2017learning}, Google~\cite{erlingsson2014rappor}, and Microsoft~\cite{ding2017collecting}. For instance, Apple~\cite{team2017learning} deploys LDP on iOS to know popular health data types for future improvement in the Health app. Google~\cite{erlingsson2014rappor} has adopted RAPPOR to Chrome Web browser, collecting data about Chrome clients from approximately 14 million respondents who have opted to send usage statistics to Google. Microsoft~\cite{ding2017collecting} integrates LDP in Windows 10 to collect the number of seconds that a user has spent using a particular app.

In social science and epidemiologic research, surveying sensitive questions such as cheat, gamble, drug or alcohol abuse, or antisocial behavior, is likely to lead to refusals or untruthful answers. To address this issue, Warner et al.~\cite{warner1965randomized} proposed a procedure as randomized response, the most classical LDP mechanism, in 1965. Since then, many mechanisms~\cite{greenberg1969unrelated, mangat1990alternative, kuk1990asking, mangat1994improved, christofides2003generalized, gjestvang2006new, esponda2006negative} for surveying sensitive questions have been proposed, two iconic ones being~\cite{greenberg1969unrelated} and~\cite{christofides2003generalized}. Simmons et al.~\cite{greenberg1969unrelated} suggested in 1967 that the level of cooperation would increase if two unrelated questions (or statements) were used. Christofides et al.~\cite{christofides2003generalized} proposed a generalization randomized response technique in 2003, where users only have to answer with numbers instead of \textit{yes} or \textit{no}.

Waseda et al.~\cite{waseda2016analyzing} evaluated Warner, Kuk~\cite{kuk1990asking}, negative survey mechanism~\cite{esponda2006negative} and its variants~\cite{aoki2012limited} by using differential privacy. They showed that these mechanisms have a tradeoff between privacy and utility, but did not compare the performance of different mechanisms nor suggest ways to improve them. Giordano et al.~\cite{giordano2012efficiency} compared the variance of the estimators of three dichotomous unrelated question mechanisms (~\cite{greenberg1969unrelated, singh2003use, perri2008modified}) under equal levels of confidentiality measures introduced by Lanke~\cite{lanke1976degree}, Leysieffer and Warner~\cite{leysieffer1976respondent}, but also did not improve the performance of variance.

With the development of big data and artificial intelligence, people can collect data not by sampling, but by surveying the population, such as in U.S. Census 2020 and the US presidential election of 2020. And with that in mind, we modified three mechanisms (Warner, Simmons and Christodfides mechanisms) for conducting surveys on sensitive topics by drawing a sample from the population without replacement. Since drawing a sample from the population without replacement can reduce the randomness of the estimated mean in our case, we improved modified Christofides mechanism by drawing cards from device without replacement. The improvement is remarkable for the variance of the improved Christofides mechanism is the smallest when the proportion of people with sensitive attributes is small if we use LDP as the same measurement of privacy leakage. In fact, the variance decreases to 28.7\% that of modified Christofides mechanism in our experiment base on a read world dataset. As a result, we find a new LDP mechanism for mean estimation with smaller variance and a better degree of cooperation from data providers.

\textbf{Contributions}. Our contributions are threefold:

(1)We analyze modified Warner, Simmons and Christodfides mechanisms by drawing samples from the population without replacement.

(2)We propose an improved Christofides mechanism with a small but effective change to modified Christofides mechanism when the proportion of people with sensitive attributes is small.

(3)We compare the variance of modified Warner, Simmons, Christodfides and the improved Christofides mechanisms by theoretical analysis and numerical simulation using LDP as the same measurement of privacy leakage. We find that the improved Christofides mechanisms has the smallest variance when the proportion of people with sensitive attributes is small, and can be used as a better LDP mechanism to replace Warner mechanism.

\textbf{Organization.} Section 2 presents the background, including the definition of LDP and the procedures of Warner, Simmons and Christofides mechanisms. Section 3 presents the
procedures and variances of modified Warner, and Simmons and Christofides mechanisms by theoretical analysis. In section 4, we introduce the improved Christofides mechanism and analyze its variance. In section 5, we analyze the variance of modified Warner, Simmons and Christofides and the improved Christofides mechanism by using LDP as the measurement of privacy leakage. In section 6, we compare the variance and minimum sample size of four mechanisms by theoretical analysis and numerical simulation. In section 7, we give the conclusion.

\section{Background}
\subsection{Local Differential Privacy}
In the LDP setting, a user perturbs the private value $x$ using an algorithm $\mathcal{A}$ and sends $\mathcal{A}(x)$ to the untrusted collector. The collector learns statistical information about users.

\begin{definition}[$\varepsilon$-LDP~\cite{duchi2013local}]
A randomization mechanism $\mathcal{A}$ satisfies $\varepsilon$-LDP, if and only if for any pair of input values $x$, $x'$ and for all randomized output $O$, it holds that
\begin{equation*}
\mathbf{P} [\mathcal{A}(x)=O] \leq e^\varepsilon \times \mathbf{P}[\mathcal{A}(x')=O] \text {.}
\end{equation*}
The notation $\mathbf{P}$ means probability.
\end{definition}

\subsection{Three Randomized Response Mechanisms}
This section briefly introduces Warner Mechanism, Simmons mechanism and Christofides mechanism.

\subsubsection{Warner Mechanism~\cite{warner1965randomized}}

Suppose that each respondent in a population belongs to either Group $A$ or $\bar A$ and that it is required to estimate the proportion of Group $A$ by survey. Each respondent is required to respond $yes$ or $no$ to one of the two statements:

(a) I am a member of group $A$.

(b) I am a member of group $\bar A$.

The respondent responds to statement (a) with probability $p$ ($p<1/2$ in this paper.) and to statement (b) with probability $1-p$ using a random device, e.g., by the toss of a (biased) coin and in the absence of the interviewer. The investigator, such as the government or school, can get the unbiased estimate of the true proportion of people with sensitive attributes according to the answers.

\subsubsection{Simmons Mechanism~\cite{greenberg1969unrelated}}

It is noted that both statements of Warner mechanism are sensitive. In Simmons mechanism, we change the second statement of Warner mechanism to an unrelated one, which leads to a higher degree of cooperation for the respondent compared with that of Warner mechanism.

Simmons mechanism uses two statements:

(a) I am a member of Group $A$.

(b) I am a member of Group $B$.

, where the unrelated statement in Group $B$ is unsensitive. For instance, the two statements posed might be:

(a) I cheated in an exam.

(b) I was born in the first half of the year.

The investigator, such as the government or school, can get the unbiased estimator of the true proportion of people with sensitive attributes according to the answers.
\subsubsection{Christofides Mechanism~\cite{christofides2003generalized}}

In Christofides mechanism, every respondent is provided with a card which produces the integers $1$, $2$,...,$L$ with proportions $p_1$,$p_2$,...,$p_L$ respectively (These proportions are not all equal). The respondent reports how far away the integer is from $L+1$ if he/she has the sensitive attributes or from $0$ if he/she does not have it. For instance, suppose that $L$ is 3 and the respondent reports the number 3. This means that either the respondent has sensitive attributes and the drawn number is 1 or the respondent does not have sensitive attributes and the drawn number is 3. Notice that when $L=2$, Christofides mechanism is substantially equivalent to Warner mechanism.

The investigator, such as the government or school, can get the unbiased estimate of the true proportion of people with sensitive attributes according to the answers.

\section{Analyzing three Modified Mechanisms}
This section analyzes modified Warner, Simmons and Christofides mechanisms. We assume the whole procedure is completed by the respondent and unobserved by the investigator with all respondents being truthful in their answers.

In~\cite{warner1965randomized}~\cite{greenberg1969unrelated} and~\cite{christofides2003generalized}, the authors assumed a random sampling from the population with replacement. However, due to the development of the Internet, people can let the number of samples and populations be equal, such as in U.S. Census 2020~\cite{UnitedStatesCensusBureau} and the US presidential election, which means sampling without replacement.

Let $x_i$ be the integer representing the sensitive attributes of respondent $i$, which is 1 if the respondent is a member of sensitive group (persons with no health insurance coverage at the time of interview), and 0 if the respondent is not a member of sensitive group (persons with health insurance coverage at the time of interview in our experiment). We have

\begin{equation*}
x_i=
\begin{cases}
1, & {\rm respondent\ }i\in{\rm\ Group\ }A \\
0, & {\rm respondent\ }i\in{\rm\ Group\ }\bar A.
\end{cases}
\end{equation*}

\subsection{Analyzing Modified Warner Mechanism}

Algorithm \ref{algorithm: modified Warner Mechanism} illustrates modified Warner Mechanism. We select from the population of size $N$ a random sample without replacement of size $N$. Let $\pi_A$ be the true proportion of being a member of group $A$ in the population.

\renewcommand{\algorithmicrequire}{\textbf{Input:}}
\renewcommand{\algorithmicensure}{\textbf{Output:}}

\begin{algorithm}
\caption{modified Warner Mechanism}
\label{algorithm: modified Warner Mechanism}
\begin{algorithmic}[1]
        \Require $x_i(i=1,2,...,N), p$
        \Ensure $X_i$
        \State {Sample a Bernoulli variable $u$ such that $P(u = 1) = p$}
        \If {$u<p$}
            \State {$X_i = x_i$}
        \Else
             \State {$X_i=1-x_i$}
        \EndIf
        \State return $X_i$
        \State Experiment with the next respondent $i+1$.
\end{algorithmic}
\end{algorithm}

Introduce random variables

\begin{equation*}
X_{i}=
\begin{cases}
1, & \text { the } i \text { th respondent answers \textit{yes}} \\
0, & \text { the } i \text { th respondent answers \textit{no}}.
\end{cases}
\end{equation*}

Random variables $X_{1}$, $X_{2}$, $\ldots$, $X_{N}$ are dependent. The number of {\it yes} answers obtained from $N$ respondents is $N_{1}=\sum_{i=1}^{N} X_{i}$.

According to the formula of total probability, the theoretical proportion of the respondent who answered $yes$ to question (a) or (b) is approximately equal to ${N_{1}}/{N}$, shown as follows,
\begin{equation*}
p \pi_{A}+(1-p)(1-\pi_{A}) \approx \frac{N_{1}}{N}.\\
\end{equation*}

The parameter $\hat{\pi_A}$  is estimated based on the indirect responses of all respondents via the estimator

\begin{equation*}
\hat{\pi_A}=\frac{\frac{N_{1}}{N}-(1-p)}{2 p-1}.\\
\end{equation*}

The estimator $\hat{\pi_A}$  is unbiased estimation according to the Law of large numbers~\cite{durrett2019probability}.

The variance of $\hat{\pi_A}$~\cite{wang2017locally} is

\begin{equation}\label{equation: Warner variance}
\begin{aligned}
\mathbf{Var}\hat{\pi_A} & =\frac{\operatorname{Var}(\sum_{i=1}^N X_i)}{N^2(2 p-1)^2} \\
& =\frac{\pi_A p(1-p)+(1-\pi_A)(1-p) p}{N(2 p-1)^2} \\
& =\frac{p(1-p)}{N(2 p-1)^2}.
\end{aligned}
\end{equation}

The variance is decreased by $[\pi_A(1-\pi_A)]/N$ compared with that in~\cite{warner1965randomized} that samples with replacement, which is

\begin{equation*}
\mathbf{Var}\hat{\pi_A}=\frac{\pi_A(1-\pi_A)}{N}+\frac{p(1-p)}{N(2 p-1)^2}.
\end{equation*}

\subsection{Analyzing Modified Simmons Mechanism}

Algorithm \ref{algorithm: modified Simmons Mechanism} illustrates modified Simmons Mechanism. We select from the population of size $N$ a random sample without replacement of size $N$. Let $\pi_A$ be the true proportion of being a member of group A in the population.

\renewcommand{\algorithmicrequire}{\textbf{Input:}}
\renewcommand{\algorithmicensure}{\textbf{Output:}}

\begin{algorithm}
\caption{modified Simmons Mechanism}
\label{algorithm: modified Simmons Mechanism}
\begin{algorithmic}[1]
        \Require $x_i(i=1,2,...,N), p, \pi_B$
        \Ensure $X_i$
        \State {Sample a Bernoulli variable $u$ such that $P(u = 1) = p$}
        \If {$u<p$}
            \State {$X_i = x_i$}
        \Else
            \State {Sample a Bernoulli variable $v$ such that $P(v = 1) = \pi_B$}
            \If {$v<\pi_B$}
                \State {$X_i = 1$}
            \Else
                \State {$X_i = 0$}
            \EndIf
        \EndIf
        \State return $X_i$
        \State Experiment with the next respondent $i+1$.
\end{algorithmic}
\end{algorithm}

Introduce random variables

\begin{equation*}
X_{i}=
\begin{cases}
1, & \text { the } i \text { th respondent answers \textit{yes}} \\
0, & \text { the } i \text { th respondent answers \textit{no}}.
\end{cases}
\end{equation*}

Random variables $X_{1}$, $X_{2}$, $\ldots$, $X_{N}$ are dependent. The number of {\it yes} answers obtained from $N$ respondents is $N_{1}=\sum_{i=1}^{N} X_{i}$.

According to the formula of total probability, the theoretical proportion of the respondent who answered $yes$ to question (a) or (b) is approximately equal to ${N_{1}}/{N}$, shown as follows,

\begin{equation*}
p \pi_{A}+(1-p) \pi_{B} \approx \frac{N_{1}}{N}.\\
\end{equation*}

The parameter ${\pi}_{A}$ is estimated based on the indirect responses of all respondents via the estimator

\begin{equation*}
\hat{\pi_A}=\frac{\frac{N_{1}}{N}-(1-p) \pi_{B}}{p}.\\
\end{equation*}

The estimator $\hat{\pi_A}$  is an unbiased estimation according to the Law of large numbers~\cite{durrett2019probability}.

The variance of $\hat{\pi_A}$ is

\begin{equation}\label{equation: Simmons mechanism variance}
\begin{aligned}
\mathbf{Var}\hat{\pi_A}= & \frac{\mathbf{Var}\left(\frac{1}{N} \sum_{i=1}^N X_i\right)}{p^2} \\
= & \frac{\pi_A\left[p+(1-p) \pi_B\right]\left[1-p-(1-p) \pi_B\right]}{N p^2}+ \\
& \frac{\left(1-\pi_A\right)\left[(1-p) \pi_B\right]\left[1-(1-p) \pi_B\right]}{N p^2} \\
= & \frac{\pi_B(1-p)-\pi_B^2(1-p)^2}{N p^2}+ \\
& \pi_A \frac{1-p-2 \pi_B(1-p)}{N p}.
\end{aligned}
\end{equation}

The variance is decreased by $[\pi_A(1-\pi_A)]/N$ compared with that in~\cite{greenberg1969unrelated} which samples with replacement,

\begin{equation*}
\begin{aligned}
\mathbf{Var}\hat{\pi_A}= & \frac{\pi_A\left(1-\pi_A\right)}{N}+\frac{\pi_B(1-p)-\pi_B^2(1-p)^2}{N p^2}+ \\
& \pi_A \frac{1-p-2 \pi_B(1-p)}{N p}.
\end{aligned}
\end{equation*}

\subsection{Analyzing Modified Christofides Mechanism}

Algorithm \ref{algorithm: modified Christofides mechanism} illustrates modified Christofides Mechanism. We draw a sample of size $N$ from the population of size $N$ without replacement. To formulate the algorithm mathematically, prepare $M$ cards, which produce integers $1$, $2$, ..., $L$ with proportions $p_1$, $p_2$, ..., $p_L$ respectively (These proportions are not all equal). Let $\pi_A$ be the true proportion of being a member of group $A$ in the population.

\renewcommand{\algorithmicrequire}{\textbf{Input:}}
\renewcommand{\algorithmicensure}{\textbf{Output:}}

\begin{algorithm}
\caption{modified Christofides Mechanism}
\label{algorithm: modified Christofides mechanism}
\begin{algorithmic}[1]
        \Require $x_i(i=1, 2, ..., N), L, p_1, p_2, ..., p_L$
        \Ensure $X_i$
        \State {prepare a device consist of $M$ cards, each card showing one of the integers $1,2,...,L$ with proportions $p_1,p_2,...,p_L$ respectively;}
        \State {randomly select a card, assume the number is $k$;}
        \If {$x_i=1$}
            \State {$X_i = k$;}
        \Else
             \State {$X_i=L+1-k$}
        \EndIf
        \State return $X_i$;
        \State Put the card back to device.
        \State Experiment with the next respondent $i+1$.
\end{algorithmic}
\end{algorithm}

For convenience, we introduce the following notations: $Y_i$ is the number he/she draws and $X_i$ is the number he/she answers.

We have

\begin{equation*}
X_{i}=
\begin{cases}
1, & \text { the } i \text { th respondent answers \textit{yes}} \\
0, & \text { the } i \text { th respondent answers \textit{no}}
\end{cases}
\end{equation*}

, where

\begin{equation*}
X_i=\left(L+1-Y_i\right) x_i+Y_i\left(1-x_i\right).
\end{equation*}

The expectation of $X$ is

\begin{equation*}
\mathbf{E} X=\mathbf{E} Y+\pi(L+1-2 \mathbf{E} Y).
\end{equation*}

Thus the estimation of $\pi_A$ is

\begin{equation*}
\hat{\pi}_A=\frac{\frac{1}{N} \sum_{i=1}^N X_i-\mathbf{E}(Y)}{L+1-2 \mathbf{E}(Y)}.
\end{equation*}

The estimator $\hat{\pi_A}$  is an unbiased estimation according to the Law of large numbers~\cite{durrett2019probability}.

\begin{theorem}[Modified Christofides mechanism variance]\label{theorem: Modified Christofides mechanism variance}
The variance of $\hat{\pi_A}$ is
\begin{equation*}
\begin{aligned}
\mathbf{Var}\hat{\pi_A} & =\frac{\mathbf{Var}\left(\sum_{i=1}^N X_i\right)}{N^2(L+1-2 \mathbf{E}(Y))^2} \\
& =\frac{\mathbf{Var}Y}{N(L+1-2 \mathbf{E}(Y))^2}.
\end{aligned}
\end{equation*}
\end{theorem}

\begin{proof}
\textit{The proof is in Appendix A.}
\end{proof}

The variance is decreased by $[\pi_A(1-\pi_A)]/N$ compared with that in~\cite{christofides2003generalized} which samples with replacement,

\begin{equation*}
\mathbf{Var}\hat{\pi_A}=\frac{\pi_A\left(1-\pi_A\right)}{N}+\frac{\operatorname{Var} Y}{N(L+1-2 \mathbf{E}(Y))^2}.
\end{equation*}

By modifying Warner, Simmons and Christofides by sampling without replacement, the variance is decreased by $[\pi_A(1-\pi_A)]/N$ compared with that in~\cite{warner1965randomized}~\cite{greenberg1969unrelated} and~\cite{christofides2003generalized}, which samples with replacement. This is our first contribution.

\section{the Improved Christofides Mechanism}

Since drawing a random sample from the population without replacement can reduce the randomness of the estimated mean in our case. What will happen if respondent draws cards without replacement? To clarify this issue, we proposed the improved Christofides mechanism, which is illustrated as Algorithm \ref{algorithm: the improved Christofides Mechanism}.

\renewcommand{\algorithmicrequire}{\textbf{Input:}}
\renewcommand{\algorithmicensure}{\textbf{Output:}}

\begin{algorithm}
\caption{the improved Christofides Mechanism}
\label{algorithm: the improved Christofides Mechanism}
\begin{algorithmic}[1]
        \Require $x_i(i=1, 2, ..., N), L, p_1, p_2, ..., p_L$
        \Ensure $X_i$
        \State {prepare a device consist of $N$ cards, each card showing one of the integers $1,2,...,L$ with proportions $p_1,p_2,...,p_L$ respectively;}
        \State {randomly select a card, assume the number is $k$;}
        \If {$x_i=1$}
            \State {$X_i = k$;}
        \Else
             \State {$X_i=L+1-k$}
        \EndIf
        \State return $X_i$;
        \State The respondent keeps the card. (Do not put the card back to device.)
        \State Experiment with the next respondent $i+1$.
\end{algorithmic}
\end{algorithm}

We select from the population of size $N$ a sample without replacement of size $N$. To formulate the algorithm mathematically, prepare $N$ cards, which produce integers $1$, $2$, ..., $L$ with proportions $p_1$, $p_2$, ..., $p_L$ respectively (These proportions are not all equal).

For convenience, we introduce the following notations: $Y_i$ is the number he/she draws and $X_i$ is the number he/she answers.

We have

\begin{equation*}
X_{i}=
\begin{cases}
1, & \text { the } i \text { th respondent answers \textit{yes}} \\
0, & \text { the } i \text { th respondent answers \textit{no}},
\end{cases}
\end{equation*}

where

\begin{equation*}
X_i=\left(L+1-Y_i\right) x_i+Y_i\left(1-x_i\right).
\end{equation*}

The expectation of $X$ is

\begin{equation*}
\mathbf{E} X=\mathbf{E} Y+\pi(L+1-2 \mathbf{E} Y).
\end{equation*}

Thus the estimation of $\pi_A$ is

\begin{equation*}
\hat{\pi}_A=\frac{\frac{1}{N} \sum_{i=1}^N X_i-\mathbf{E}(Y)}{L+1-2 \mathbf{E}(Y)}.
\end{equation*}

The estimator $\hat{\pi_A}$  is an unbiased estimation according to the Law of large numbers~\cite{durrett2019probability}.

\begin{theorem}[the improved Christofides mechanism variance]\label{theorem: the improved Christofides mechanism variance}
The variance of $\hat{\pi_A}$ is
\begin{equation}
\begin{aligned}
\mathbf{Var}\hat{\pi_A} & =\frac{\mathbf{Var}\left(\sum_{i=1}^N X_i\right)}{N^2(L+1-2 \mathrm{E}(Y))^2} \\
& =\frac{4 \pi_A\left(1-\pi_A\right) \mathbf{Var}(Y)}{(N-1)(L+1-2 \mathbf{E}(Y))^2}.
\end{aligned}
\end{equation}
\end{theorem}

\begin{proof}
\textit{The proof is in Appendix B.}
\end{proof}

Since the estimation of $\pi_A$ for modified Christofides mechanisms and the improved Christofides mechanism are all unbiased estimations, we compare the utility of the two mechanisms by calculating the difference between the variance of modified Christofides ($\mathbf{Var}_{IC}$) and the improved Christofides mechanism ($\mathbf{Var}_{MC}$).

\begin{theorem}[Comparison Modified Christofides and the Improved Christofides Mechanisms by Theoretical Analysis]\label{theorem: Comparison Modified Christofides and the Improved Christofides Mechanisms by Theoretical Analysis}

The comparison can be summarized as follows:

(\romannumeral1.) When proportion $\pi_A$ is in interval $\left(0, 1/2-1/2\sqrt{N}\right) \cup(1/2+1/2\sqrt{N}, 1)$, we get that $\mathbf{Var_{IC}}$ is smaller than $\mathbf{Var_{MC}}$.

(\romannumeral2.) When proportion $\pi_A$ is in interval $(1/2-1/2\sqrt{N}, 1/2+1/2$\\$\sqrt{N})$, we get that $\mathbf{Var_{IC}}$ is larger than $\mathbf{Var_{MC}}$.

\end{theorem}

\begin{proof}

\begin{equation*}
\begin{aligned}
& \mathbf{Var_{IC}}\hat{\pi_A}-\mathbf{Var_{MC}}\hat{\pi_A}\\
& =\left[\frac{4 \pi_A\left(1-\pi_A\right)}{N-1}-\frac{1}{N}\right] \frac{\mathbf{Var} Y}{(L+1-2 \mathbf{E} Y)^2} \\
& =\left[\frac{-4 N \pi_A^2+4 N \pi_A-N+1}{N(N-1)}\right] \frac{\mathbf{V a r} Y}{(L+1-2 \mathbf{E} Y)^2} \\
& =-\frac{4(1/2+{1}/{2\sqrt{N}}-\pi_A)(1/2-{1}/{2\sqrt{N}}-\pi_A) \mathbf{Var} Y}{(N-1)(L+1-2 \mathbf{E} Y)^2}.
\end{aligned}
\end{equation*}

As long as $\pi_A<1/2-{1}/{2\sqrt{N}}$ or $\pi_A>1/2+{1}/{2\sqrt{N}}$, we have

\begin{equation*}
\begin{aligned}
& \mathbf{Var}_{IC}\hat{\pi_A}<\mathbf{Var}_{MC}\hat{\pi_A}.\qedhere
\end{aligned}
\end{equation*}
\end{proof}

Because $N$ is a large number, the interval $[1/2-{1}/{2\sqrt{N}}, 1/2+{1}/{2\sqrt{N}}]$ is small. For example, when $N=10^4$, the interval is $[0.495, 0.505]$. Especially, since ${1}/{2\sqrt{N}} \leq 1/4$(for $N \geq 3)$, as long as $\pi_A<0.25$ or $\pi_A>0.75$, we have $\mathbf{Var_{IC}}<\mathbf{Var_{MC}}$ for any $N(N>3)$.

Moreover, we will prove in section \ref{section: Comparing four Mechanisms using LDP} that the variance of the improved Christofides mechanism is smaller than that of modified Warner and Simmons mechanisms under certain assumption, in using LDP as a measurement of privacy leakage. The assumption is do satisfied usually in the real world. We improved Christofides mechanism by sampling cards without replacement. This is our second contribution.

\section{Analyzing four Mechanisms Variances using LDP}

In this section, we analyzed the variance of four mechanisms by theoretical analysis using LDP as the same measurement of privacy leakage.

\subsection{Modified Warner Mechanism using LDP}

\theoremstyle{plain}
\begin{theorem}[Warner mechanism]\label{theorem: Warner mechanism}

Modified Warner mechanism satisfies $\varepsilon-$differential privacy, where $\varepsilon=\ln [(1-p)/p]$ (without losing of generality, suppose that $p<1/2$.).
\end{theorem}

\begin{proof} We consider four cases:

Case 1: The respondent has sensitive attributes, and the probability of answering \textit{yes} under this condition is $p$.

Case 2: The respondent has sensitive attributes, and the probability of answering \textit{no} under this condition is $1-p$.

Case 3: The respondent does not have sensitive attributes, and the probability of answering \textit{yes} under this condition is $1-p$.

Case 4: The respondent does not have sensitive attributes, and the probability of answering \textit{no} under this condition is $p$.

For any two respondents with the same answer, the maximum ratio of proportions is $(1-p)/p$.
\end{proof}

Substitute $\varepsilon=\ln [(1-p)/p]$ and $p<{1}/{2}$ into equation (\ref{equation: Warner variance}). The variance $\mathbf{Var_{MW}}$ can be written as

\begin{equation*}
\mathbf{Var}\hat{\pi_A}=\frac{e^{\varepsilon}}{N\left(e^{\varepsilon}-1\right)^{2}}.
\end{equation*}

\subsection{Modified Simmons Mechanism using LDP}

\begin{theorem}[Simmons mechanism]\label{theorem: Simmons mechanism}
Modified Simmons mechanism satisfies $\varepsilon-$differential privacy, where
\begin{equation}\label{equation: Simmons mechanism epsilon}
\varepsilon=\begin{cases}
	\ln \frac{p+(1-p) \pi_{B}}{(1-p) \pi_{B}} &, \pi_{B} \leq \frac{1}{2}\\
	\ln \frac{p+(1-p)\left(1-\pi_{B}\right)}{(1-p)\left(1-\pi_{B}\right)} &, \pi_{B}>\frac{1}{2}.
\end{cases}
\end{equation}
\end{theorem}

\begin{proof} We consider four cases:

Case 1: The respondent has sensitive attributes, and the probability of answering \textit{yes} under this condition is $p+(1-p)\pi_B$.

Case 2: The respondent has sensitive attributes, and the probability of answering \textit{no} under this condition is $(1-p)(1-\pi_B)$.

Case 3: The respondent does not have sensitive attributes, and the probability of answering \textit{yes} under this condition is $(1-p)\pi_B$.

Case 4: The respondent does not have sensitive attributes, and the probability of answering \textit{no} under this condition is $p+(1-p)(1-\pi_B)$.

For any two respondents with the same answer, the maximum ratio of proportions is
\begin{equation*}
\varepsilon=\begin{cases}
	\ln \frac{p+(1-p) \pi_{B}}{(1-p) \pi_{B}} &, \pi_{B} \leq \frac{1}{2}\\
	\ln \frac{p+(1-p)\left(1-\pi_{B}\right)}{(1-p)\left(1-\pi_{B}\right)} &, \pi_{B}>\frac{1}{2}. \hfill \qedhere
\end{cases}
\end{equation*}
\end{proof}

\begin{theorem}[Modified Simmons mechanism variance versus epsilon]\label{theorem: Modified Simmons mechanism variance versus epsilon}
When $\pi_B=1/2$, the variance $\mathbf{Var_{MS}}$ obtains its minimum as following
\begin{equation*}
\mathbf{Var}\hat{\pi_A}=\frac{e^{\varepsilon}}{N\left(e^{\varepsilon}-1\right)^2} .
\end{equation*}
\end{theorem}
 which is the same with that of modified Warner mechanism.
\begin{proof}
\textit{The proof is in Appendix C.}
\end{proof}

\subsection{Modified Christofides Mechanism using LDP}

\begin{theorem}[Modified Christofides mechanism]\label{theorem: Christofides mechanism}

Modified Christofides mechanism satisfies $\varepsilon$-differential privacy, where
\begin{equation*}
\varepsilon=\max \left(\ln \frac{p_{L+1-k}}{p_k}, k=1,2, \cdots, L\right).
\end{equation*}
\end{theorem}

\begin{proof} We consider four cases:

Case 1: The respondent has sensitive attributes and draws $k$. The probability of answering $L+1-k$ under this condition is $p_{k}$.

Case 2: The respondent has sensitive attributes and draws $L+1-k$. The probability of answering $k$ under this condition is $p_{L+1-k}$.

Case 3: The respondent does not have sensitive attributes and draws $k$. The probability of answering $k$ under this condition is $p_{k}$.

Case 4: The respondent does not have sensitive attributes and draws $L+1-k$. The probability of answering $L+1-k$ under this condition is $p_{L+1-k}$.

For any two respondents with the same answer, the maximum ratio of probabilities is

\begin{equation*}
\varepsilon=\max \left(\ln \frac{p_{L+1-k}}{p_k}, k=1,2, \cdots, L\right).
\end{equation*}
\end{proof}

We fix number $L=3$, so the privacy budget $\varepsilon$ is

\begin{equation*}
\varepsilon=\max \left\{\ln\frac{p_1}{p_3}, \ln\frac{p_3}{p_1}\right\} .
\end{equation*}

\begin{theorem}[Modified Christofides mechanism variance versus epsilon]\label{theorem: Modified Christofides mechanism variance versus epsilon}

For the case $L=3$, the variance of the modified Christofides mechanism $\mathbf{Var_{MC}}$ arrives its minimum

\begin{equation*}
\mathbf{Var}\hat{\pi_A}=\frac{1}{4 N}\left[\frac{(e^\epsilon+1)^2}{(e^\epsilon-1)^2\left(1-p_2\right)}-1\right].
\end{equation*}
\end{theorem}

under the situation that proportions
\begin{equation*}
p_1=\left(1-p_2\right) /\left(e^{\varepsilon}+1\right), p_3=\left[e^{\varepsilon}\left(1-p_2\right)\right] /\left(e^{\varepsilon}+1\right)
\end{equation*}

or

\begin{equation*}
p_1=\left[e^{\varepsilon}\left(1-p_2\right)\right] /\left(e^{\varepsilon}+1\right), p_3=\left(1-p_2\right) /\left(e^{\varepsilon}+1\right).
\end{equation*}

\begin{proof}
\textit{The proof is in Appendix D.}
\end{proof}

Notice that when proportion $p_2=0$, modified Christofides mechanism is substantially equivalent to modified Warner mechanism.

\subsection{Improved Christofides Mechanism using LDP}

\begin{theorem}[the improved Christofides mechanism]\label{theorem: the improved Christofides mechanism}
Suppose that any respondent does not know the cards others drew and the numbers others answered. Then the improved Christofides mechanism satisfies $\varepsilon$-differential privacy, where
\begin{equation*}
\varepsilon=\max \left(\ln \frac{p_{L+1-k}}{p_k}, k=1,2, \cdots, L\right).
\end{equation*}
\end{theorem}

\begin{proof} We consider four cases:

Case 1: The respondent has sensitive attributes and draws $k$. The probability of answering $L+1-k$ under this condition is $p_{k}$.

Case 2: The respondent has sensitive attributes and draws $L+1-k$. The probability of answering $k$ under this condition is $p_{L+1-k}$.

Case 3: The respondent does not have sensitive attributes and draws $k$. The probability of answering $k$ under this condition is $p_{k}$.

Case 4: The respondent does not have sensitive attributes and draws $L+1-k$. The probability of answering $L+1-k$ under this condition is $p_{L+1-k}$.

For any two respondents with the same answer, the maximum ratio of probabilities is

\begin{equation*}
\varepsilon=\max \left(\ln \frac{p_{L+1-k}}{p_k}, k=1,2, \cdots, L\right).
\end{equation*}
\end{proof}

We fix number $L=3$, so the privacy budget $\varepsilon$ is

\begin{equation*}
\varepsilon=\max \left\{\ln\frac{p_1}{p_3}, \ln\frac{p_3}{p_1}\right\} .
\end{equation*}

\begin{theorem}[the improved Christofides variance versus epsilon]\label{theorem: the improved Christofides variance versus epsilon}
For the case $L=3$, the variance of the improved Christofides mechanism $\mathbf{Var_{IC}}$ arrives its minimum

\begin{equation*}
\mathbf{Var}\hat{\pi_A}=\frac{\pi_A\left(1-\pi_A\right)}{(N-1)}\left[\frac{\left(e^\epsilon+1\right)^2}{\left(e^\epsilon-1\right)^2\left(1-p_2\right)}-1\right]
\end{equation*}

under the situation that proportions
\begin{equation*}
p_1=\left(1-p_2\right) /\left(e^{\varepsilon}+1\right), p_3=\left[e^{\varepsilon}\left(1-p_2\right)\right] /\left(e^{\varepsilon}+1\right)
\end{equation*}

or

\begin{equation*}
p_1=\left[e^{\varepsilon}\left(1-p_2\right)\right] /\left(e^{\varepsilon}+1\right), p_3=\left(1-p_2\right) /\left(e^{\varepsilon}+1\right).
\end{equation*}
\end{theorem}

\begin{proof}
\textit{The proof is similar with that in Appendix D}.
\end{proof}

\section{Comparing four Mechanisms using LDP}\label{section: Comparing four Mechanisms using LDP}

Since the estimation of $\pi_A$ for modified Warner/Simmons, Christofides mechanisms and the improved Christofides mechanism are all unbiased estimations, we compare variance of the four mechanisms by theoretical analysis and numerical simulation.

\subsection{Comparing four Mechanisms using LDP by Theoretical Analysis}
This section compares the utility of these four mechanisms by calculating variances under the same privacy budget $\varepsilon$.

\begin{theorem}[Comparison four Mechanisms using LDP by Theoretical Analysis]\label{theorem: Comparison four Mechanisms using LDP by Theoretical Analysis}

The comparison can be summarized as follows:

(\romannumeral1.) When proportion $\pi_A$ is in interval $\left(0, \pi_{A_1}\right) \cup\left(\pi_{A_2}, 1\right)$, $\mathbf{Var_{IC}}<\mathbf{Var_{MW/MS}}<\mathbf{Var_{MC}}$. $\pi_{A_1}$ and $\pi_{A_2}$ are shown in equation (\ref{equation: pi_A1 and pi_A2}).

(\romannumeral2.) When proportion $\pi_A$ is in interval $(\pi_{A_1}, 1/2-1/2\sqrt{N}) \cup( 1/2+1/2\sqrt{N}, \pi_{A_2})$, $\mathbf{Var_{MW/MS}}<\mathbf{Var_{IC}}<\mathbf{Var_{MC}}$.

(\romannumeral3.) When proportion $\pi_A$ is in interval $(1/2-1/2\sqrt{N}, 1/2+1\\/2\sqrt{N})$, $\mathbf{Var_{MW/MS}} <\mathbf{Var_{MC}}<\mathbf{Var_{IC}}$.

\end{theorem}
\begin{proof}

Firstly, according to theorem \ref{theorem: Comparison Modified Christofides and the Improved Christofides Mechanisms by Theoretical Analysis}, when proportion $\pi_A$ is in interval $(1/2-1/2\sqrt{N}, 1/2+1/2\sqrt{N})$, we have $\mathbf{Var_{IC}}>\mathbf{Var_{MC}}$, while $\mathbf{Var_{IC}}<\mathbf{Var_{MC}}$ when proportion $\pi_A$ is in interval $(0, 1/2-1/2\sqrt{N}) \cup (1/2+1/2\sqrt{N}, 1)$.

Secondly, we compare the variance of modified Christofides mechanism $\mathbf{Var_{MC}}$ with that of the of modified Warner/Simmons mechanisms ($\mathbf{Var_{MW/MS}}$).

Since

\begin{equation*}
\begin{aligned}
& \mathbf{Var_{MC}}\hat{\pi_A} - \mathbf{Var_{MW/MS}}\hat{\pi_A} \\
& =\frac{\left(e^{\varepsilon}+1\right)^2 p_2}{4 N\left(e^{\varepsilon}-1\right)^2\left(1-p_2\right)}, \\
& \geq 0.
\end{aligned}
\end{equation*}

the variance of modified Christofides mechanism $\mathbf{Var_{MC}}$ is large than that of modified Warner/Simmons mechanisms $\mathbf{Var_{MW/MS}}$.

Then, we compare the variance of the improved Christofides mechanism $\mathbf{Var_{IC}}$ with that of modified Warner/Simmons mechanisms $\mathbf{Var_{MW/MS}}$. The difference between the variances of the improved Christofides mechanism and modified Warner/Simmons mechanisms is

\begin{equation*}
\begin{aligned}
& \mathbf{Var_{IC}}\hat{\pi_A}-\mathbf{Var_{MW/MS}}\hat{\pi_A} \\
&= \frac{\pi_A\left(1-\pi_A\right)}{N-1}\left[\frac{\left(e^{\varepsilon}+1\right)^2}{\left(e^{\varepsilon}-1\right)^2\left(1-p_2\right)}-1\right]-\frac{e^{\varepsilon}}{N\left(e^{\varepsilon}-1\right)^2} \\
&=\frac{\left[4 e^{\varepsilon}+\left(e^{\varepsilon}-1\right)^2 p_2\right] \pi_A\left(1-\pi_A\right)-(N-1)\left(1-p_2\right) e^{\varepsilon}}{N(N-1)\left(e^{\varepsilon}-1\right)^2\left(1-p_2\right)}
\end{aligned}
\end{equation*}

Let

\begin{equation*}
a=4 e^{\varepsilon}+\left(e^{\varepsilon}-1\right)^2 p_2, c=(N-1)\left(1-p_2\right) e^{\varepsilon}
\end{equation*}

We have

\begin{equation*}
\begin{aligned}
& \mathbf{ Var_{IC}}\hat{\pi_A}-\mathbf{Var_{MW/MS}}\hat{\pi_A}\\
& =\frac{-a \pi_A^2+a \pi_A-c}{N(N-1)\left(e^{\varepsilon}-1\right)^2\left(1-p_2\right)} \\
& =\frac{-a\left(\pi_A-\pi_{A_1}\right)\left(\pi_A-\pi_{A_2}\right)}{N(N-1)\left(e^{\varepsilon}-1\right)^2\left(1-p_2\right)}
\end{aligned}
\end{equation*}

\begin{equation*}
\begin{aligned}
& \pi_{A_1}=\frac{a-\sqrt{\Delta}}{2 a} \\
& \pi_{A_2}=\frac{a+\sqrt{\Delta}}{2 a} \\
& \Delta=a^2-4 a c.
\end{aligned}
\end{equation*}

\begin{equation*}
\begin{aligned}
& \pi_{A_1}=\frac{1}{2}-\frac{1}{2} \sqrt{1-\frac{4 c}{a}} \\
& =\frac{1}{2}-\frac{1}{2} \sqrt{1-\frac{4 e^{\varepsilon}(N-1)\left(1-p_2\right)}{N\left[4 e^{\varepsilon}+p_2\left(e^{2 \varepsilon}-2 e^{\varepsilon}+1\right)\right]}} \\
& =\frac{1}{2}-\frac{1}{2} \sqrt{1-\frac{(N-1)\left(1-p_2\right)}{N\left[1+\frac{p_2}{4}\left(e^{\varepsilon}+\frac{1}{e^{\varepsilon}}-2\right)\right]}}.
\end{aligned}
\end{equation*}

\begin{equation*}
\begin{aligned}
& \pi_{A_2}=\frac{1}{2}+\frac{1}{2} \sqrt{1-\frac{4 c}{a}} \\
& =\frac{1}{2}+\frac{1}{2} \sqrt{1-\frac{4 e^{\varepsilon}(N-1)\left(1-p_2\right)}{N\left[4 e^{\varepsilon}+p_2\left(e^{2 \varepsilon}-2 e^{\varepsilon}+1\right)\right]}} \\
& =\frac{1}{2}+\frac{1}{2} \sqrt{1-\frac{(N-1)\left(1-p_2\right)}{N\left[1+\frac{p_2}{4}\left(e^{\varepsilon}+\frac{1}{e^{\varepsilon}}-2\right)\right]}}.
\end{aligned}
\end{equation*}

Let

\begin{equation*}
d=e^{\varepsilon}+\frac{1}{e^{\varepsilon}}-2.
\end{equation*}

So

\begin{equation}\label{equation: pi_A1 and pi_A2}
\begin{aligned}
& \pi_{A_1}=\frac{1}{2}-\frac{1}{2} \sqrt{1-\frac{N-1}{N} \cdot \frac{1-p_2}{1+\frac{p_2d}{4}}} \\
& \pi_{A_2}=\frac{1}{2}+\frac{1}{2} \sqrt{1-\frac{N-1}{N} \cdot \frac{1-p_2}{1+\frac{p_2d}{4}}}.
\end{aligned}
\end{equation}

As long as $\pi_{A}<\pi_{A_1}$ or $\pi_{A}>\pi_{A_2}$, we have

\begin{equation*}
\begin{aligned}
\mathbf{ Var_{IC}}\hat{\pi_A}-\mathbf{Var_{MW/MS}}\hat{\pi_A}<0.\qedhere
\end{aligned}
\end{equation*}

\end{proof}

The improved Christofides mechanism performs worst in interval $[\pi_{A_1},\pi_{A_2}]$. The interval $[\pi_{A_1},\pi_{A_2}]$ depends on $N$, $\varepsilon$ and $p_2$. Because the interval is symmetric about 0.5, we calculate the interval length under different proportions $p_2$ and privacy budgets $\varepsilon$ after fix $N=10^4$. The results are reported in Table \ref{table: interval length}. The interval $[\pi_{A_1},\pi_{A_2}]$ is usually small.

\begin{table}[htbp]
\caption{The length of interval $[\pi_{A_1},\pi_{A_2}]$ under different proportion $p_2$ and privacy budget $\varepsilon$.}
    \centering
    \begin{tabular}{|c|c|c|c|c|}
    \hline
        \diagbox{$p_2$}{$\varepsilon$} & 0.01 & 0.05 & 0.25 & 0.5 \\ \hline
        0.01 & 0.100 & 0.101 & 0.101 & 0.104  \\ \hline
        0.05 & 0.224 & 0.224 & 0.225 & 0.230  \\ \hline
    \end{tabular}
    \label{table: interval length}
\end{table}

%

\subsection{Comparing four Mechanisms using LDP by Numerical Simulation}

To plot the curve of variance versus privacy budget $\varepsilon$ by numerical simulation for four mechanisms, we fix sample size $N=10^2$ and proportion $\pi_A=0.1$ for four mechanisms, fix proportion $\pi_B=0.5$ for modified Simmons mechanism, and proportion $p_2=0.5$ for modified Christofides and the improved Christofides mechanisms. Then experiment according to algorithm \ref{algorithm: modified Warner Mechanism} to \ref{algorithm: the improved Christofides Mechanism}. We report the results averaging 10000 runs. The theoretical analysis and numerical simulation result of four mechanisms is shown in Fig.\ref{figure: modifiedWarnerSimmonmodifiedChristofidesandtheimprovedChristofides}. The theoretical analysis value is verified by numerical simulation value. It shows that for these four mechanisms, the larger the privacy budget $\varepsilon$ is, the smaller the variance is.

\begin{figure}[htbp]
\centering
\includegraphics[width=3.5in]{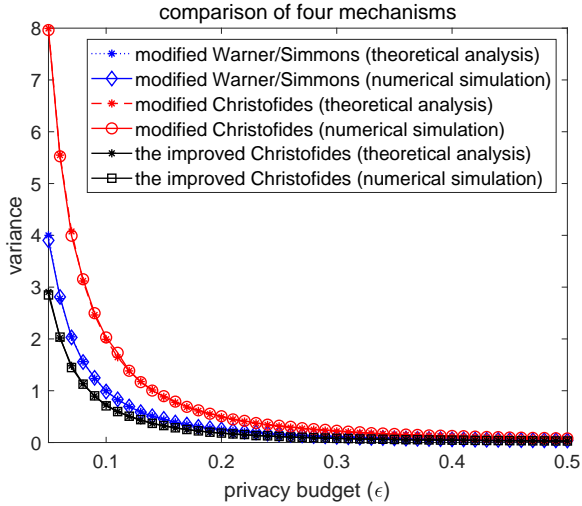}
\caption{The variance versus privacy budget $\varepsilon$ of four mechanisms when fix sample size $N=10^2$ and proportion $\pi_A=0.1$ (by theoretical analysis and numerical simulation).}
\label{figure: modifiedWarnerSimmonmodifiedChristofidesandtheimprovedChristofides}
\end{figure}

We then report curves that variance versus privacy budget $\varepsilon$ of four mechanisms by numerical simulation under different conditions in theorem \ref{theorem: Comparison four Mechanisms using LDP by Theoretical Analysis}. We fix sample size $N=9$ for four mechanisms, fix proportions $\pi_B=0.5$ for modified Simmons mechanism, fix proportion $p_2=0.36$ for modified Christofides and the improved Christofides mechanisms, and set proportion $\pi_A \in \left\{1/9, 2/9, 4/9 \right\}$. The results are reported in Fig. \ref{compare_versus_pi_A=0.11} to \ref{compare_versus_pi_A=0.44}.

Fig. \ref{compare_versus_pi_A=0.11} shows that when proportion $\pi_A$ is in interval $\left(0, \pi_{A_1}\right) \cup\left(\pi_{A_2}, 1\right)$, the improved Christofides mechanism is the best, and modified Warner/Simmons mechanism is suboptimal, while modified Christofides mechanism is the worst.

Fig. \ref{compare_versus_pi_A=0.22} shows that when proportion $\pi_A$ is in interval $(\pi_{A_1},  1/2 \\ -1/2\sqrt{N}) \cup\left(1/2+1/2\sqrt{N}, \pi_{A_2}\right)$, modified Warner/Simmons mechanism is the best, and the improved Christofides mechanism is suboptimal, while modified Christofides mechanism is the worst.

Fig. \ref{compare_versus_pi_A=0.44} shows that when proportion $\pi_A$ is in interval $(1/2-1/2\sqrt{N}, 1/2+1/2\sqrt{N})$, modified Christofides mechanism is the best, and modified Warner/Simmons mechanism is suboptimal, while the improved Christofides mechanism is the worst.

In the real world, the proportion of people with sensitive attributes is usually small. This means the improved Christofides mechanism is usually the best method at most time.

Actually, in section \ref {section: Comparing four Mechanisms using LDP on a Real World Dataset}, we experiment based on a read world dataset.
\begin{figure}[htbp]
\centering
\includegraphics[width=3.5in]{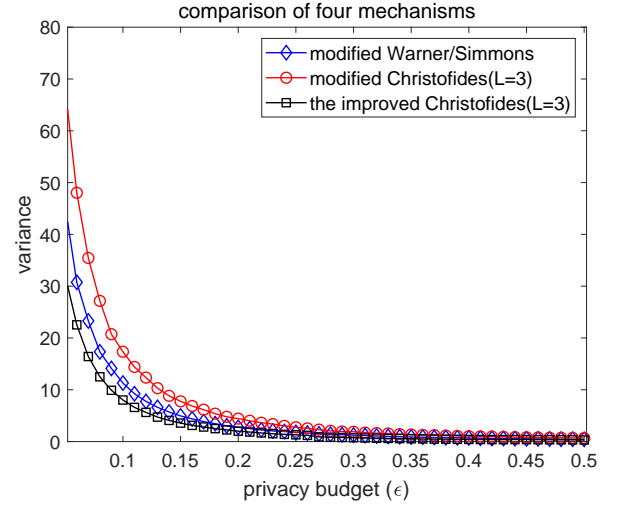}
\caption{The variance versus privacy budget $\varepsilon$ of four mechanisms when fix sample size $N=9$ and proportion $\pi_A=1/9$ (by numerical simulation).}
\label{compare_versus_pi_A=0.11}
\end{figure}

\begin{figure}[htbp]
\centering
\includegraphics[width=3.5in]{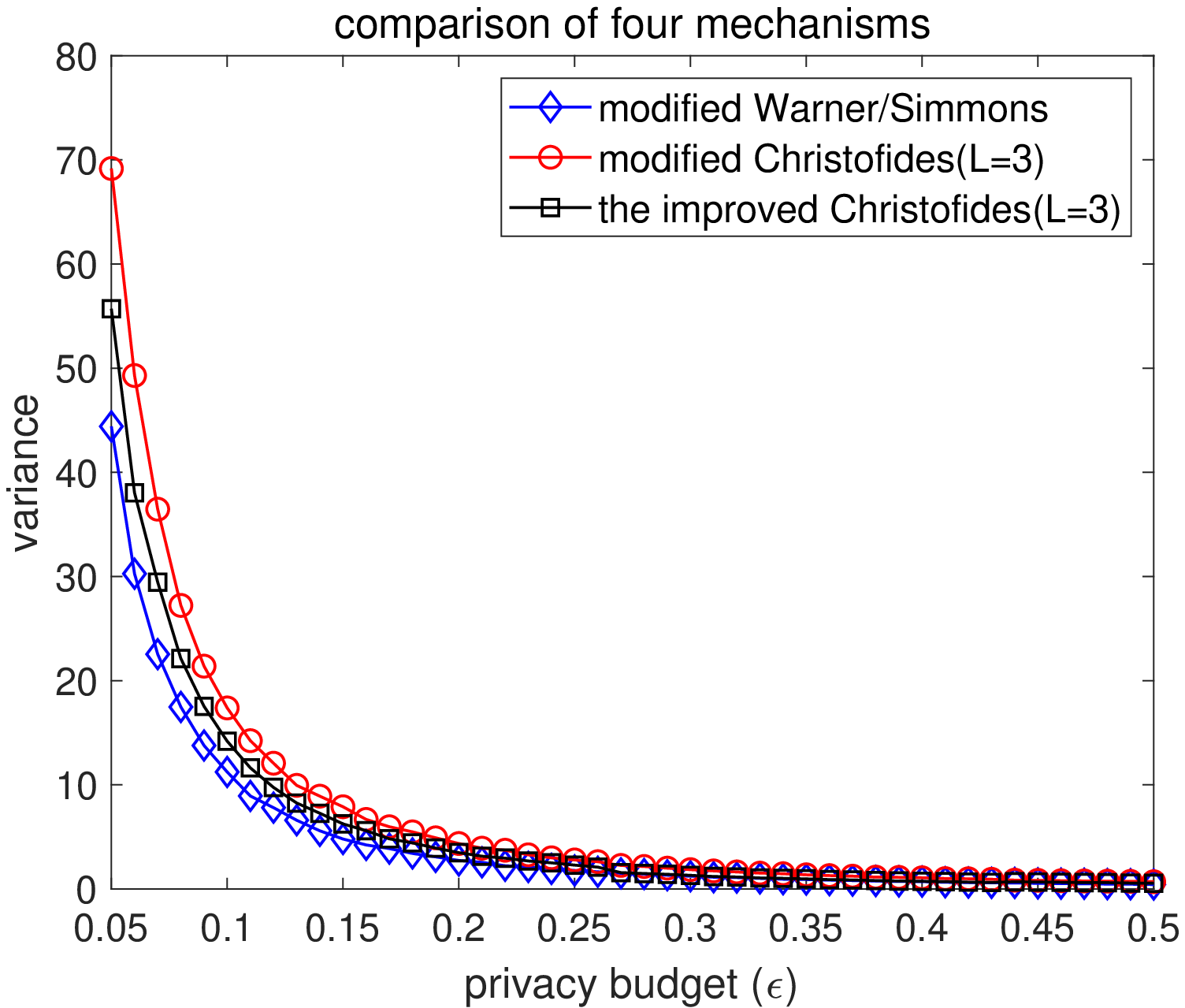}
\caption{The variance versus privacy budget $\varepsilon$ of four mechanisms when fix sample size $N=9$ and proportion $\pi_A=2/9$ (by numerical simulation).}
\label{compare_versus_pi_A=0.22}
\end{figure}

\begin{figure}[htbp]
\centering
\includegraphics[width=3.5in]{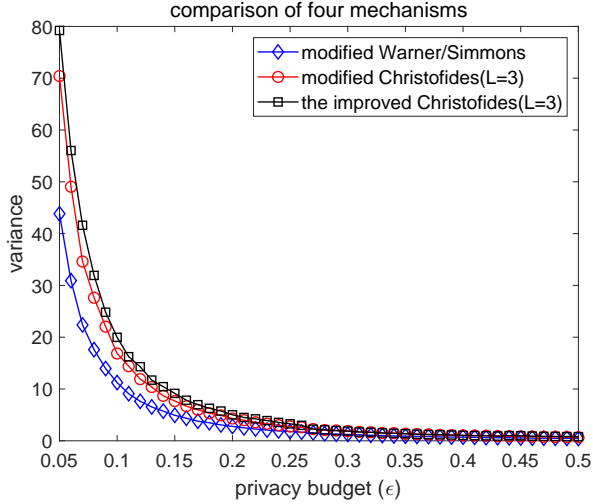}
\caption{The variance versus privacy budget $\varepsilon$ of four mechanisms when fix sample size $N=9$ and proportion $\pi_A=4/9$ (by numerical simulation).}
\label{compare_versus_pi_A=0.44}
\end{figure}

\subsection{Comparing four Mechanisms using LDP on a Real World Dataset}\label{section: Comparing four Mechanisms using LDP on a Real World Dataset}
\textbf{Dataset.}\textit{HCOVANY Dataset~\cite{HCOVANY2021}.} We conduct our experiments on an open real world dataset. HCOVANY indicates whether persons had any health insurance coverage during the interview, as measured by employer-provided insurance(HINSEMP), privately purchased insurance (HINSPUR), Medicare (HINSCARE), Medicaid or other governmental insurance (HINSCAID), TRICARE or other military care (HINSTRI), or Veterans Administration-provided insurance (HINSVA). The Census Bureau does not consider the respondents with health insurance coverage if their only coverage is from Indian Health Services (HINSIHS), as IHS policies are not always comprehensive. Codes of 2 indicate that a person is covered (either directly or through another household member's policy) by the given type of insurance; codes of 1 indicate that a person is not covered. We select the subset of \textit{HCOVANY Dataset} in our experiment (the health insurance coverage of $1\%$ of total population in USA in 2021, the sampling size $N=3252599$).

\textbf{Privacy Budget $\varepsilon$.} According to Harvard's list~\cite{Harvard} of varying levels of sensitivity for datasets and reasonable privacy loss parameters for each level, we set privacy budget $\varepsilon$ from 0.05 to 0.5, since health insurance coverage information that could cause risk of material harm to individuals if disclosed.

We fix proportion $p_2=0.01$ for modified Christofides mechanism and the improved Christofides mechanism. Then experiment according to algorithm \ref{algorithm: modified Warner Mechanism} to \ref{algorithm: the improved Christofides Mechanism}. We report the results averaging $10^4$ runs. The comparison of variance of these four mechanisms versus the privacy budget $\varepsilon$ are shown in Fig. \ref{figure: compare based on a real dataset}.

\begin{figure}[htbp]
\centering
\includegraphics[width=3.5in]{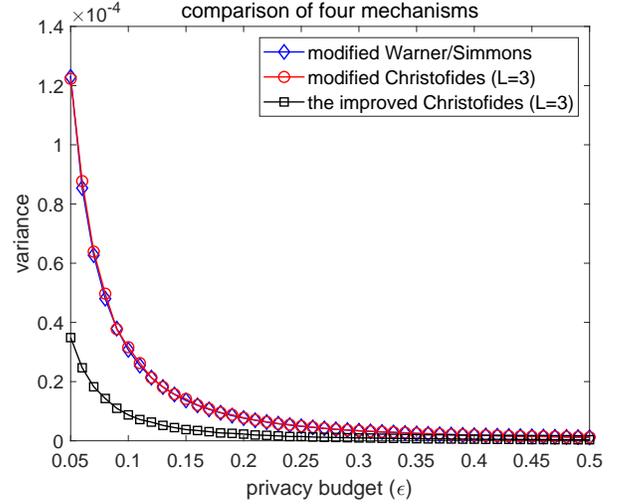}
\caption{The comparison of variance versus privacy budget $\varepsilon$ of four mechanisms based on HCOVANY dataset.}
\label{figure: compare based on a real dataset}
\end{figure}

The variance of modified Warner mechanism is the same with that of modified Simmons mechanism. The variance of modified Christofides mechanism is slightly larger than that of modified Warner and Simmons mechanism. The variance of the improved Christofides mechanism is the smallest of mechanisms. In fact, we decrease the variance to $28.7\%$ of modified Christofides mechanism's variance based on HCOVANY dataset.

\begin{equation*}
\begin{aligned}
& \frac{\mathbf{Var}{_{IC}\hat{\pi_{A}}}}{\mathbf{Var}{_{MC}\hat{\pi_{A}}}}= \frac{4 N \pi_A\left(1-\pi_A\right)}{N-1} \\
& =\frac{4 \times 3252599 \times 0.0778\times(1-0.0778)}{3252599-1}\\
& =28.7\%.
\end{aligned}
\end{equation*}

\subsection{Comparing the Minimum Sample Size of the four Mechanisms}

Since an important task for statisticians is to determine the minimum number of samples $N$ under given constraints, we compare the minimum sampling size $N$ for the given proportion $\pi_A=0.1$, variance $\mathbf{Var}\hat{\pi_A}=0.1$ and privacy budget $\varepsilon$ of four mechanisms(modified Warner, Simmons, Christofides ($p_2=0.01$) and the improved Christofides mechanisms ($p_2=0.01$)). The result is shown in Table \ref{table: N}. The result shows that the larger the privacy budget $\varepsilon$ is, the smaller the sample size $N$ is for the four mechansims. The minimum number of samples $N$ of modified Warner and Simmons mechanisms are the same. The minimum number of samples $N$ of modified Christofides mechanism is slightly larger than that of modified Warner/Simmons mechanisms. The minimum number of samples $N$ of the improved Christofides mechanism is the smallest one.

\begin{table}[htbp]
\caption{The minimum sampling size $N$ for given proportion $\pi_A=0.1$, variance $\mathbf{Var}\hat{\pi_A}=0.1$ and privacy budget $\varepsilon$ of four mechanisms.}
    \centering
    \begin{tabular}{|c|c|c|c|c|}
    \hline
        \diagbox{Mechanisms}{Privacy Budget} & 0.01 & 0.05 & 0.25 & 0.5  \\ \hline
        modified Warner/Simmons & 100000 & 4000 & 160 & 40  \\ \hline
        modified Christofides & 101011 & 4040 & 161 & 40  \\ \hline
        the improved Christofides & 36365 & 1456 & 59 & 16 \\ \hline
    \end{tabular}
    \label{table: N}
\end{table}

Especially, if we do not know that the proportion belongs to sensitive group when using the improved Christofides mechanism, for example, under constraints of variance $\mathbf{Var} \leq 0.1$, privacy budget $\varepsilon=0.01$ and choosing proportion $p_2=0.01$, since the variance get the maximum when $\pi_A=0.5$, sample size $N$ should be at least 101011($N>101010.3$).

Since the second statement of modified Simmons mechanism is unrelated and insensitive, the degree of cooperation of the respondents with sensitive attributes is higher than that of modified Warner mechanism. On account of the respondent is only the requirement to report numbers in modified Christofides and the improved Christofides mechanisms. The degree of cooperation of the respondents with sensitive attributes is higher than that of modified Warner and Simmons mechanism. In the improved Christofides mechanism, none of the respondent knows the numbers others drew and the numbers others answered, which is proper in the real world. The freedom to know numbers others drew and numbers others answered is sacrificed in exchange for smaller variance under the same privacy budget $\varepsilon$.

We compare the variance and minimum sample size of modified Warner, Simmons, Christofides and the improved Christofides mechanisms using LDP as the common measurement of privacy leakage by theoretical analysis, numerical simulation and experiment based on a real world dataset. As a result, we find the improved Christofides mechanism for LDP framework. This is our third contribution.

\section{Conclusion}

In conclusion, we have modified Warner, Simmons and Christofides mechanism by drawing samples without replacement, thus reducing the variance of each mechanism. We then improve modified Christofides mechanism by sampling cards without replacement. Considering theoretical analysis, numerical simulation and experiment based on a real world dataset, we compared modified Warner, Simmons, Christofides and the improved Christofides mechanism using LDP as the measurement of privacy leakage. We found that the variance of modified Christofides mechanism was the smallest when the proportion of people with sensitive attributes $\pi_A < \pi_{A_1}$ or $\pi_A > \pi_{A_2}$, which is usually proper in the real world, and that the improved Christofides mechanism can be used to replace the Warner mechanism in sensitive issue investigation and LDP scenarios, as it offers a better trade-off between privacy budget, variance, and the degree of cooperation.

\section*{Appendices}
\appendix

\section{Proof of theorem \ref{theorem: Modified Christofides mechanism variance}}

\setcounter{theorem}{0}
\renewcommand{\thetheorem}{3.\arabic{theorem}}
\begin{theorem}[Modified Christofides mechanism variance]\label{theorem: ModifiedChristofides mechanism variance}
The variance of $\hat{\pi_A}$ is
\begin{equation*}
\begin{aligned}
\mathbf{Var}\hat{\pi_A} & =\frac{\mathbf{Var}\left(\sum_{i=1}^N X_i\right)}{N^2(L+1-2 \mathbf{E}(Y))^2} \\
& =\frac{\mathbf{Var}Y}{N(L+1-2 \mathbf{E}(Y))^2}.
\end{aligned}
\end{equation*}
\end{theorem}

\begin{proof}
Firstly, we have

\begin{equation*}
\begin{aligned}
\mathbf{Var}\left(\sum_{i=1}^NX_i\right)= & N\pi_A\left[\sum_{k=1}^Lk^2p_{L+1-k}-\left(\sum_{k=1}^Lkp_{L+1-k}\right)^2\right]+ \\
& N(1-\pi_A)\left[\sum_{k=1}^Lk^2p_{k}-\left(\sum_{k=1}^Lkp_k\right)^2\right] \\
\end{aligned}
\end{equation*}

\begin{equation*}
\begin{aligned}
= & N\pi_A\left[\sum_{k=1}^L(L+1-k)^2p_k-\sum_{k=1}^Lk^2p_k  \right]+ \\
& N\pi_A\left[-(\sum_{k=1}^L(L+1-k)p_k)^2+(\sum_{k=1}^Lkp_k)^2\right]+\\
& N\left[\sum_{k=1}^Lk^2p_{k}-\left(\sum_{k=1}^Lkp_k\right)^2 \right]\\
= & N\pi_A\left[\sum_{k=1}^L((L+1)^2-2(L+1)k)p_k\right]- \\
& N\pi_A\left[-(\sum_{k=1}^L(L+1-k)p_k)^2+(\sum_{k=1}^Lkp_k)^2\right]+\\
& N\pi_A\left[(\sum_{k=1}^L((L+1)p_k)(\sum_{k=1}^L((L+1-2k)p_k)\right]\\
= & N\left[\sum_{k=1}^Lk^2p_{k}-\left(\sum_{k=1}^Lkp_k\right)^2 \right].
\end{aligned}
\end{equation*}

Thus
\begin{equation*}
\begin{aligned}
\mathbf{Var}\hat{\pi_A} &= \frac{\mathbf{Var}\left(\frac{1}{N}\sum_{i=1}^NX_i\right)}{(L+1-2E(Y))^2} \\
     &= \frac{\mathbf{Var}Y}{N(L+1-2E(Y))^2}. \qedhere
\end{aligned}
\end{equation*}
\end{proof}

\section{Proof of theorem \ref{theorem: the improved Christofides mechanism variance}}

\setcounter{theorem}{0}
\renewcommand{\thetheorem}{4.\arabic{theorem}}
\begin{theorem}[the improved Christofides mechanism variance]\label{theorem: the improved Christofides mechanism variance}
The variance of $\hat{\pi_A}$ is
\begin{equation*}
\begin{aligned}
\mathbf{Var}\hat{\pi_A} & =\frac{\mathbf{Var}\left(\sum_{i=1}^N X_i\right)}{N^2(L+1-2 \mathrm{E}(Y))^2} \\
& =\frac{4 \pi_A\left(1-\pi_A\right) \mathbf{Var}(Y)}{(N-1)(L+1-2 \mathbf{E}(Y))^2}.
\end{aligned}
\end{equation*}
\end{theorem}

\begin{proof}
Since the estimator of $\pi_A$
\begin{equation*}
\hat {\pi_A}=\frac{\frac{1}{N}\sum_{i=1}^N X_i-\mathbf{E} Y}{(L+1-2\mathbf{E} Y)}.
\end{equation*}
with the variance

\begin{equation*}
\begin{aligned}
\mathbf{Var} \hat{\pi_A}&=\mathbf{Var}\left( \frac{\frac{1}{N}\sum_{i=1}^N X_i-\mathbf{E} Y}{L+1-2\mathbf{E} Y}\right)\\
&=\frac{1}{N^2(L+1-2\mathbf{E} Y)^2}\mathbf{Var}\sum_{i=1}^N X_i
\end{aligned}
\end{equation*}

,where

\begin{equation*}
\begin{aligned}
\mathbf{Var}\sum_{i=1}^N X_i&=\mathbf{E}(\sum_{i=1}^N X_i)^2-(\mathbf{E} \sum_{i=1}^NX_i)^2\\
  &=\sum_{i=1}^N \mathbf{E} (X_i^2)+\sum_{i\neq j}\mathbf{E}\left( X_iX_j\right)-(\mathbf{E} \sum_{i=1}^NX_i)^2.\\
\end{aligned}
\end{equation*}

In the improved Christofides mechanism, $X_i$ and $X_j$ are no longer i.i.d. case, which yields difficulty in the proof.
Here, both the participants and cards are sampled without replacement, which gives us dependent $Y_i$ and $Y_j$, $x_i$ and $x_j$, but $Y_i$ and $x_i$ are independent.

Firstly,

\begin{equation}\label{equation: the improved Christofides mechanism variance firstly}
\begin{aligned}
& \sum_{i=1}^N \mathbf{E}\left(X_i^2\right)=\sum_{i=1}^N \mathbf{E}\left[\left(L+1-Y_i\right)^2 x_i^2+Y_i^2 x_i^2\right] \\
& =N\left[\mathbf{E}(L+1-Y)^2 \mathbf{E} x_i^2+\mathbf{E} Y^2 \mathbf{E}\left(1-x_i\right)^2\right] \\
& =N\left[(L+1)^2+\mathbf{E} Y^2-2(L+1) \mathbf{E} Y\right] \mathbf{E} x_i^2+\\
& N \mathbf{E} Y^2 \mathbf{E}\left(1-2 x_i+x_i^2\right) \\
& =N \mathbf{E} Y^2-2 N(L+1) \pi_A \mathbf{E} Y+N(L+1)^2 \pi_A.
\end{aligned}
\end{equation}

Secondly,

For $i \neq j$
\begin{equation*}
\begin{aligned}
\mathbf{E} X_iX_j &=\mathbf{E} \left((L+1-Y_i)x_i+Y_i(1-x_i)\right)\cdot\\
&\left((L+1-Y_j)x_j+Y_j(1-x_j)\right),\\
&=\mathbf{E} (L+1-Y_i)(L+1-Y_j)\mathbf{E} x_ix_j\\
&+\mathbf{E} Y_iY_j\mathbf{E} (1-x_i)(1-x_j)\\
&+2\mathbf{E} Y_i(L+1-Y_j)\mathbf{E}(1-x_i)x_j.
\end{aligned}
\end{equation*}

Because we have
\begin{equation*}
\begin{aligned}
\mathbf{E} Y_iY_j &=\sum_{m=1}^N \sum_{n=1}^N m n \mathbf{P}(Y_i=m,\,Y_j=n)\\
&=\sum_{m=1}^N m\left(\sum_{k\neq m}^N n p_m\frac{p_nN}{N-1}+mp_m\frac{p_mN-1}{N-1}\right)\\
&=\frac{N}{N-1}\mathbf{E} Y\sum_{m=1}^N mp_m-\frac{1}{N-1}\sum_{m=1}^N m^2p_l\\
&=(\mathbf{E} Y)^2-\frac{1}{N-1}\mathbf{Var} Y.
\end{aligned}
\end{equation*}
and

\begin{equation*}
\begin{aligned}
&\mathbf{E} x_i=\mathbf{P}(x_i=1)=\pi_A\\
&\mathbf{E} x_ix_j=\mathbf{P}(x_i=1,\, x_j=1)=\frac{C_{N\pi_A}^2}{C_N^2}\\
&\mathbf{E} (1-x_i)(1-x_j)(i \neq j)=\mathbf{P}(x_i=0,\, x_j=0)=\frac{C_{N-N\pi_A}^2}{C_N^2}\\
&\mathbf{E} (1-x_i)x_j(i \neq j)=\mathbf{P}(x_i=0,\, x_j=1)=\frac{C_{N\pi_A}^1C_{N-N\pi_A}^1}{A_N^2}.
\end{aligned}
\end{equation*}

We can get

\begin{equation}\label{equation: the improved Christofides mechanism variance secondly}
\begin{aligned}
\mathbf{E} X_i X_j & =\mathbf{E}[(L+1)^2-Y_i(L+1)+Y_i Y_j] \mathbf{E} x_i x_j+ \mathbf{E} Y_i Y_j \\
& \mathbf{E}(1-x_i)(1-x_j)+ \mathbf{E}(L+1-Y_i) Y_j \mathbf{E} x_i(1-x_j)+\\
& \mathbf{E} Y_i(L+1-Y_j) \mathbf{E}(1-x_i) x_j \\
& =\mathbf{E} Y_i Y_j[\mathbf{E} x_i x_j+\mathbf{E}(1-x_i)(1-x_j)-\mathbf{E} x_i(1-x_j)-\\
& \mathbf{E}(1-x_i) x_j] + {\mathbf{E} Y[-2(L+1) \mathbf{E} x_i x_j+(L+1) \mathbf{E} x_i} \\
& {(1-x_j)+ (L+1) \mathbf{E}(1-x_i) x_j]} +(L+1)^2 \mathbf{E} x_i x_j \\
& ={[N \pi_A(N \pi_A-1)+(N-N \pi_A)(N-N \pi_A-1)-} \\
& {2N+\pi_A(N-N \pi_A)] \mathbf{E} Y_i Y_j + [-2(L+1) N \pi_A(N \pi_A-}\\
& {1)+2(L+1) N \pi_A(N-N \pi_A)] \mathbf{E} Y + (L+1)^2 N \pi_A}\\
& {(N \pi_A-1)}\\
& ={[4 N^2 \pi_A^2-4 N^2 \pi_A+N^2-N] \mathbf{E} Y_i Y_j+2(L+1)} \\
& {N \pi_A[1+N-2 N \pi_A] \mathbf{E} Y+(L+1)^2 N \pi_A(N \pi_A-1)} \\
& ={[4 N^2 \pi_A^2-4 N^2 \pi_A+N^2-N][(\mathbf{E} Y)^2-\frac{1}{N-1}} \\
& {{\mathbf{Var} Y]} + 2(L+1) N \pi_A[1+N-2 N \pi_A] \mathbf{E} Y+(L+1)^2} \\
& {N \pi_A(N \pi_A-1).}
\end{aligned}
\end{equation}

Thirdly,
\begin{equation}\label{equation: the improved Christofides mechanism variance thirdly}
\begin{aligned}
\mathbf{E}(\sum_{i=1}^N X_i)^2 & =N^2[\mathbf{E} Y+\pi_A(L+1-2 \mathbf{E} Y)]^2 \\
& =N^2[(L+1) \pi_A+(1-2 \pi_A) \mathbf{E} Y]^2 \\
& ={N^2[(L+1)^2 \pi_A^2+(1-2 \pi_A)^2(\mathbf{E} Y)^2+2(L+1) \pi_A} \\
& {(1-2 \pi_A) \mathbf{E} Y].}
\end{aligned}
\end{equation}

Finally, combine equation(\ref{equation: the improved Christofides mechanism variance firstly}), (\ref{equation: the improved Christofides mechanism variance secondly}) and (\ref{equation: the improved Christofides mechanism variance thirdly}), the variance of the estimator $\hat {\pi_A}$ is
\begin{equation*}
\begin{aligned}
\mathbf{Var} \hat{\pi_A}& =\frac{1}{N^2(L+1-2\mathbf{E} Y)^2}\mathbf{Var}\sum_{i=1}^N X_i\\
&=\frac{4\pi_A(1-\pi_A)\mathbf{Var} Y}{(N-1)(L+1-2\mathbf{E} Y)^2}. \qedhere
\end{aligned}
\end{equation*}
\end{proof}

\section{Proof of theorem \ref{theorem: Modified Simmons mechanism variance versus epsilon}}

\setcounter{theorem}{2}
\renewcommand{\thetheorem}{5.\arabic{theorem}}

\begin{theorem}[Modified Simmons mechanism variance versus epsilon]\label{theorem: Modified Simmons mechanism variance versus epsilon}
When $\pi_B=1/2$, the variance $\mathbf{Var_{\pi_A}}$ obtains minimum as following
\begin{equation*}
\mathbf{Var}\hat{\pi_A}=\frac{e^{\varepsilon}}{N\left(e^{\varepsilon}-1\right)^2}.
\end{equation*}
\end{theorem}

\begin{proof}
Substitute (\ref{equation: Simmons mechanism epsilon}) into (\ref{equation: Simmons mechanism variance}). The variance can be written as
\begin{equation*}
\begin{aligned}
(1) \textit{When} \quad \pi_B \leq \frac{1}{2}, \\
\mathbf{Var}\hat{\pi_A}=& \frac{\left(\frac{1}{e^{\varepsilon}-1}+\pi_A\right)\left(1+\frac{1}{\left(e^{\varepsilon}-1\right) \pi_B}\right)}{N}-\frac{1}{N\left(e^{\varepsilon}-1\right)^2} \\
& -\frac{\left(e^{\varepsilon}+1\right) \pi_A}{N\left(e^{\varepsilon}-1\right)} \\
\min \mathbf{Var}\hat{\pi_A}=& \mathbf{Var}\hat{\pi_A}_{|_{\pi_B=\frac{1}{2}}}=\frac{e^{\varepsilon}}{N\left(e^{\varepsilon}-1\right)^2}.
\end{aligned}
\end{equation*}

\begin{equation*}
\begin{aligned}
(2) \textit{When} \quad \pi_B>\frac{1}{2}, \\
\mathbf{Var}\hat{\pi}_A=& \frac{\left(\frac{1}{e^{\varepsilon}-1}+\pi_A\right)\left(1+\frac{1}{\left(e^{\varepsilon}-1\right)\left(1-\pi_B\right)}\right)}{N}-\frac{1}{N\left(e^{\varepsilon}-1\right)^2} \\
& -\frac{\left(e^{\varepsilon}+1\right) \pi_A}{N\left(e^{\varepsilon}-1\right)} \\
\min \mathbf{Var}\hat{\pi_A}=& \mathbf{Var}\hat{\pi_A}_{\pi_B=\frac{1}{2}}=\frac{e^{\varepsilon}}{N\left(e^{\varepsilon}-1\right)^2}.
\end{aligned}
\end{equation*}

So the variance of the estimator $\hat {\pi_A}$ is

\begin{equation*}
	\begin{aligned}
        \mathbf{Var}\hat{\pi_A}=&\frac{e^{\varepsilon}}{N\left(e^{\varepsilon}-1\right)^{2}}.\qedhere
	\end{aligned}
\end{equation*}
\end{proof}

\section{Proof of theorem \ref{theorem: Modified Christofides mechanism variance versus epsilon}}

\setcounter{theorem}{4}
\renewcommand{\thetheorem}{5.\arabic{theorem}}
\begin{theorem}[Modified Christofides mechanism variance versus epsilon]\label{theorem: Modified Christofides mechanism variance versus epsilon}
For the case $L=3$, the variance of the modified Christofides mechanism $\mathbf{Var_{MC}}$ arrives its minimum

\begin{equation*}
\mathbf{Var}\hat{\pi_A}=\frac{1}{4 N}\left[\frac{(e^\epsilon+1)^2}{(e^\epsilon-1)^2\left(1-p_2\right)}-1\right]
\end{equation*}

under the situation that proportions
\begin{equation*}
p_1=\left(1-p_2\right) /\left(e^{\varepsilon}+1\right), p_3=\left[e^{\varepsilon}\left(1-p_2\right)\right] /\left(e^{\varepsilon}+1\right)
\end{equation*}

or

\begin{equation*}
p_1=\left[e^{\varepsilon}\left(1-p_2\right)\right] /\left(e^{\varepsilon}+1\right), p_3=\left(1-p_2\right) /\left(e^{\varepsilon}+1\right).
\end{equation*}
\end{theorem}

\begin{proof}

The expectation of $Y$ is
\begin{equation*}
\begin{aligned}
\mathbf{E} Y & =p_1+p_2+3\left(1-p_1-p_2\right) \\
& =3-2 p_1-p_2.
\end{aligned}
\end{equation*}

The variance of $Y$ is
\begin{equation*}
\begin{aligned}
\mathbf{Var} Y& =p_1+4 p_2+9\left(1-p_1-p_2\right)-\left(3-2 p_1-p_2\right)^2 \\
& =4 p_1+p_2-4 p_1^2-p_2^2-4 p_1 p_2.
\end{aligned}
\end{equation*}

\begin{equation*}
\begin{aligned}
\mathbf{Var}\hat{\pi}_A& =\frac{\mathbf{Var} Y}{N(3+1-2 \mathbf{E} Y)^2} \\
& =\frac{4 p_1+p_2-4 p_1^2-p_2^2-4 p_1 p_2}{N\left[3+1-2\left(3-2 p_1-p_2\right)\right]^2} \\
& =\frac{4 p_1+p_2-4 p_1^2-p_2^2-4 p_1 p_2}{N\left(4 p_1+2 p_2-2\right)^2}\\
& =f(p_1,p_2).
\end{aligned}
\end{equation*}

Now let's minimize the function $f(p_1,p_2)$.

Since

\begin{equation*}
\begin{aligned}
\frac{\partial f\left(p_1, p_2\right)}{\partial p_2} & =\frac{\left(1-4 p_1-2 p_2\right)\left(4 p_1+2 p_2-2\right)^2}{N\left(4 p_1+2 p_2-2\right)^4}- \\
& \frac{4\left(4 p_1+p_2-4 p_1^2-p_2^2-4 p_1 p_2\right)\left(4 p_1+2 p_2-2\right)}{N\left(4 p_1+2 p_2-2\right)^4} \\
& =\frac{\left(p_1+\frac{1-p_2}{2}\right)\left(p_1-\frac{1-p_2}{2}\right)}{N\left(2 p_1+p_2-1\right)^4} \\
& \neq 0.
\end{aligned}
\end{equation*}

the function $f(p_1,p_2)$ has no extremum. The maximum or minimum points must be boundary points.

Since

\begin{equation*}
\frac{1}{e^{\varepsilon}} p_1 \leq p_3=1-p_1-p_2 \leq e^{\varepsilon} p_1
\end{equation*}

thus

\begin{equation*}
\frac{1-p_2}{e^{\varepsilon}+1} \leq p_1 \leq \frac{e^{\varepsilon}}{e^{\varepsilon}+1}(1-p_2).
\end{equation*}

Because the function for the two bounding points

\begin{equation*}
\begin{aligned}
\left.f\left(p_1, p_2\right)\right|_{p_1=\frac{1-p_2}{e^{\varepsilon}+1}} & =\frac{4 \frac{1-p_2}{e^{\varepsilon}+1}+p_2-4\left(\frac{1-p_2}{e^{\varepsilon}+1}\right)^2-p_2^2-4 \frac{1-p_2}{e^{\varepsilon}+1} p_2}{N\left[4 \frac{1-p_2}{e^{\varepsilon}+1}+2 p_2-2\right]^2} \\
& =\frac{1}{4 N}\left[\frac{\left(e^{\varepsilon}+1\right)^2}{\left(e^{\varepsilon}-1\right)^2\left(1-p_2\right)}-1\right] \\
\end{aligned}
\end{equation*}

\begin{equation*}
\begin{aligned}
f(p_1, p_2)|_{p_1=\frac{e^{\varepsilon}(1-p_2)}{e^{\varepsilon}+1}} & ={\frac{4 \frac{e^{\varepsilon}(1-p_2)}{e^{\varepsilon}+1}+p_2-4(\frac{e^{\varepsilon}(1-p_2)}{e^{\varepsilon}+1})^2}{N[\frac{e^{\varepsilon}}{e^{\varepsilon}+1}(1-p_2)+2 p_2-2]^2}-}\\
& {\frac{p_2^2+4 \frac{e^{\varepsilon}}{e^{\varepsilon}+1}(1-p_2) p_2}{N[\frac{e^{\varepsilon}}{e^{\varepsilon}+1}(1-p_2)+2 p_2-2]^2}}. \\
& =\frac{1}{4 N}[\frac{(e^{\varepsilon}+1)^2}{(e^{\varepsilon}-1)^2(1-p_2)}-1]
\end{aligned}
\end{equation*}

and

\begin{equation*}
f\left(p_1, p_2\right)|_{p_1=\frac{1-p_2}{2}} \rightarrow \infty.
\end{equation*}

Therefore

\begin{equation*}
\min f\left(p_1, p_2\right)=   f\left(\frac{1-p_2}{e^{\varepsilon}+1}, p_2\right)=f\left(\frac{e^{\varepsilon}(1-p_2)}{e^{\varepsilon}+1}, p_2\right).
\end{equation*}

So
\begin{equation*}
\mathbf{Var} \hat{\pi_A}=\frac{1}{4 N}\left[\frac{\left(e^{\varepsilon}+1\right)^2}{\left(e^{\varepsilon}-1\right)^2\left(1-p_2\right)}-1\right]. \qedhere
\end{equation*}
\end{proof}

\begin{acks}
This research has been supported by National Key R\&D Program of China (2022YFB4501500, 2022YFB4501504) and Zhejiang Lab (2022NF0AC01). Any opinions, findings and conclusions or recommendations expressed in this material are those of the authors and do not reflect the views of the funding agencies.
\end{acks}


\bibliographystyle{ACM-Reference-Format}
\balance
\bibliography{VLDB2023}

\end{document}